\documentclass[a4paper,UKenglish]{lipics-v2019}

\usepackage{microtype}

\usepackage[toc,page]{appendix} 

\usepackage{booktabs} 

\usepackage[utf8]{inputenc}

\usepackage{subcaption}
\usepackage{float}

\usepackage{tikz}
\usetikzlibrary{arrows,shapes,trees, arrows.meta}

\usepackage[gen]{eurosym} 

\makeatletter
\@fleqnfalse
\makeatother

\hideLIPIcs
\nolinenumbers

\title{Network-Aware Strategies in Financial Systems}

\author{Pál András Papp}{ETH Zürich, Switzerland}{apapp@ethz.ch}{}{}
\author{Roger Wattenhofer}{ETH Zürich, Switzerland}{wattenhofer@ethz.ch}{}{}

\authorrunning{P.A. Papp and R. Wattenhofer}

\Copyright{Pál András Papp and Roger Wattenhofer}



\begin{document}

\maketitle

\begin{abstract}
We study the incentives of banks in a financial network, where the network consists of debt contracts and credit default swaps (CDSs) between banks. One of the most important questions in such a system is the problem of deciding which of the banks are in default, and how much of their liabilities these banks can pay. We study the payoff and preferences of the banks in the different solutions to this problem. We also introduce a more refined model which allows assigning priorities to payment obligations; this provides a more expressive and realistic model of real-life financial systems, while it always ensures the existence of a solution.

The main focus of the paper is an analysis of the actions that a single bank can execute in a financial system in order to influence the outcome to its advantage. We show that removing an incoming debt, or donating funds to another bank can result in a single new solution that is strictly more favorable to the acting bank. We also show that increasing the bank's external funds or modifying the priorities of outgoing payments cannot introduce a more favorable new solution into the system, but may allow the bank to remove some unfavorable solutions, or to increase its recovery rate. Finally, we show how the actions of two banks in a simple financial system can result in classical game theoretic situations like the prisoner's dilemma or the dollar auction, demonstrating the wide expressive capability of the financial system model.
\end{abstract}

\newpage
\setcounter{page}{1}

\section{Introduction}

The world's financial system is a complex network where financial institutions such as banks are connected via various kinds of financial contracts. If some financial institutions go bankrupt, then others might suffer as well; the financial network might experience a ripple effect. Two of the most common financial contracts are (i) debt contracts (some bank owes a specific amount of money to another bank) and (ii) Credit Default Swaps (CDSs). A CDS is a simple financial derivative where the payment obligation depends on the defaulting of another bank in the system. The combination of debt contracts and CDSs turns out to provide a simple and yet expressive model, which is able to capture a wide range of interesting phenomena in real-life financial markets.

Given a set of banks and a set of payment obligations between these banks, one of the most natural questions is to decide which of the banks can fulfill these obligations, and which of them cannot, and hence are in default. The problem of deciding what portion of obligations banks can fulfill is known as the \textit{clearing problem}. One can easily encounter a situation when this problem has multiple different solutions in a financial system. It is natural to study how much the individual banks prefer these solutions, i.e. what is their payoff in specific solutions of the system.

In this paper we study the problem from the point of view of a single bank $v$. We analyze whether some simple actions of $v$ can improve its situation in the network. In a financial system, the complex interconnection between the banks can easily result in situations where banks can achieve a better outcome in surprising and somewhat counterintuitive ways. For example, being on the receiving end of a debt contract is generally considered beneficial, because the bank obtains payment from this contract. However, in a system with debts and CDSs, it is also possible that if a bank $v$ nullifies a debt contract as a creditor, then (through a number of intermediate steps in the network) this results in an even higher total payoff for $v$. Such phenomena are crucial to understand, since if banks indeed execute these actions to obtain a better outcome, then these opportunities will determine how the financial system changes and evolves in the future.

We begin with a description of the financial system model recently developed by Schuldenzucker \textit{et. al.} \cite{base1}, which serves as the base model for our findings. We then introduce a more refined version of this model which also assigns \textit{priorities} to each contract, and assumes that banks have to fulfill their payment obligations in the order defined by these priorities. We show that besides being more expressive and realistic, this augmented model still ensures the existence of a solution.

Our main contribution is an analysis of various different actions that banks in the system can execute in order to increase their final payoff when the system is cleared. We first show that by removing an incoming debt (partially or entirely) or by donating extra funds to another bank, a bank might be able to increase its payoff. We then show that investing more external assets or reprioritizing its outgoing payments can also allow a bank to influence the system. However, these actions do not allow a bank to introduce more favorable new solutions, but they can allow the bank to remove unfavorable solutions from the system, or increase its own recovery rate.

Finally, we present some simple examples where two banks try to influence the financial system simultaneously, resulting in situations that are identical to the classical prisoner's dilemma or dollar auction game. This suggests that financial systems in this model can exhibit very rich behavior, and if two or more banks execute these actions simultaneously, this can easily lead to complex game-theoretic settings.

\section{Related Work}

Numerous studies on the properties of financial systems are directly or indirectly based on the financial system model introduced by Eisenberg and Noe in \cite{model1}. This model only assumes simple debt contracts between banks. Different studies have later also extended this model with default costs \cite{model2}, cross-ownership relations \cite{cross1, cross2} or so-called covered CDSs \cite{coveredCDS}. The related literature has studied the propagation of shocks in many different variants of these models \cite{prop1, prop2, prop3, prop4, prop5, prop6}.

One disadvantage of these models is that they can only describe \textit{long positions} of banks on each other, meaning that a worse situation for one bank is always worse (or the same) for any other bank. For example, if a bank is unable to pay its debt, then its creditor receives less money, and it might not be able to pay its debts either. This already enables the model to capture many interesting phenomena, e.g. how a small shock causes a ripple effect in the network. However, long connections imply that there is a solution in these systems which is simultaneously the best for all banks. As such, the models cannot represent the opposing interests of banks in many real-world situations, and thus these models are not so interesting from a game-theoretic point of view.

On the other hand, a more realistic model was recently introduced by Schuldenzucker, Seuken and Battiston \cite{base1}; we assume this model of financial systems in our paper. Besides debt contracts, this new model also allows credit default swaps between banks, which are essentially financial derivatives where banks are betting on the default of another bank. CDSs are a prominent kind of derivative that played a significant role in the 2008 financial crisis \cite{CDS3}; as such, they have been studied in various works in the financial literature \cite{CDS1, CDS2, CDS4}. While the model still remains relatively simple with these two kind of contracts, it now also allows us to model \textit{short positions}, when it is more favorable for a bank if another bank is worse off. This increases the expressive power of the model dramatically, allowing us to capture a wide range of properties of practical financial systems.

The work of Schuldenzucker \textit{et. al.} analyzes their model from a complexity-theoretic perspective. The authors show that in the base variant of this model, each system has at least one solution; however, if we also assume so-called default costs, then some systems might not have a solution at all. In case of default costs, they also describe sufficient conditions for the existence of a solution. Their follow-up work shows that it is computationally hard to decide if a solution exists, and also to find or approximate a solution of the system \cite{base2}.

However, to our knowledge, the model has not been analyzed from a game-theoretic perspective before. Our paper aims to lay the foundations of such an analysis, by evaluating a variety of simple (and yet realistic) actions that allow nodes to influence the network due to the presence of short positions. Since banks often have conflicting interests in these systems, these actions can easily lead to interesting game-theoretical dilemmas.

The only similar game-theoretic analysis we are aware of is the recent work of Bertschinger \textit{et. al.} \cite{gametheo}, set in the original model of Eisenberg and Noe. Instead of having institutional rules for payment obligations in case of default, \cite{gametheo} assumes that banks can freely select the order of paying their outgoing debts, or even decide to make partial payments in some contracts. The paper discusses the properties of Nash-Equilibria and Social Optima in this setting. While this has a connection to our observations in Section \ref{sec:reprio}, we analyze the results of such actions in a significantly more complex model with CDSs.

In general, measuring the sensitivity or complexity of a financial network has also been exhaustively studied \cite{graph2, graph3, prop3, prop4}. The topic also has a major importance for financial authorities in practice, who regularly conduct stress tests to analyze real-world financial systems. The clearing problem, in particular, also plays an important role in the European Central Bank's stress test framework \cite{ECB}, for example. 

\section{Financial system model}

The model introduced by \cite{base1} describes a financial network as a set of \textit{banks} (i.e. nodes), denoted by $V$, with different kinds of financial contracts (i.e. directed edges) between specific pairs of banks. Banks in our examples are usually denoted by $u$, $v$ or $w$. Every bank in the system has a predefined amount of \textit{external assets}, denoted by $e_v$ for bank $v$.

\subsection{Debt and CDS contracts}

We assume that each contract in the system is between two specific banks $u$ and $v$. A contract obliges $u$ (the debtor) to pay a specific amount of money to bank $v$ (the creditor), either unconditionally or based on a specific event. The amount of payment obligation in the contract is the \textit{weight} (in financial terms: the notional) of the contract.

While these contracts might be connected to earlier transactions between the banks (e.g. a loan offered by $v$ to $u$ in the past which results in a debt contract from $u$ to $v$ in the present), we assume that these initial payments are implicitly represented in the external assets of banks, and thus the external assets and the contracts together provide all the necessary information to describe the current state of the system.

The outgoing contracts of bank $v$ altogether specify a given amount of total payment obligations for $v$. If $v$ is unable to fulfill all these obligations, then we say that $v$ is \textit{in default}. In this case, we are interested in the portion of liabilities that $v$ is still able to pay, known as the \textit{recovery rate} of $v$ and denoted by $r_v$. The definition shows that we always have $r_v \in [0,1]$, and $v$ is in default exactly if $r_v<1$. The recovery rates of all banks is represented together in a recovery rate vector $r \in [0,1]^V$.

The model allows two kinds of contracts between banks in the system. In case of a simple \textit{debt} contract, $u$ has to pay a specific amount to $v$ unconditionally, i.e. in any case. On the other hand, \textit{credit default swaps} (\textit{CDSs}) are ternary financial contracts, made in reference to a third bank $w$ known as the \textit{reference entity}. A CDS describes a conditional debt which only requires $u$ to pay a specific amount to $v$ if $w$ is in default. In particular, if the weight of the CDS is $\delta$ and the recovery rate of $w$ is $r_w$, then the CDS incurs a payment obligation of $\delta \cdot (1-r_w)$ from $u$ to $v$.

In practice, CDSs often describe an insurance policy on debt contracts for the creditor bank. If $v$ is the creditor of a debt coming from $w$, and $v$ suspects that $w$ might go into default and thus will be unable to pay some of its debt, then $v$ can enter into a CDS as a creditor with some other bank $u$ in the system, in reference to $w$. If $w$ indeed defaults and cannot pay its liabilities to $v$, then $v$ instead receives some payment from $u$. Nonetheless, there could be other reasons for banks to enter CDS contracts, e.g. speculative bets about future developments in the market.

\subsection{Assets and liabilities}

Since payment obligations in CDSs depend on the recovery rate of other banks, the assets and liabilities of a bank are defined as a function of the vector $r$. The \textit{liability} of $u$ towards $v$ is the sum of payment obligations from all simple debt contracts and CDSs, i.e.
\[ l_{u,v}(r) = c_{u,v} +\sum_{w \in V} c_{u,v}^w \cdot (1-r_w), \]
where $c_{u,v}$ denotes the weight of the simple debt from $u$ to $v$, and $c_{u,v}^w$ denotes the weight of the CDS from $u$ to $v$ with reference to $w$ (understood as 0 if the contracts do not exist). The total liabilities of $u$ is then the sum of liabilities to all other banks, i.e.
\[ l_u(r) = \sum_{v \in V} l_{u,v}(r). \]
In contrast to this, the actual \textit{payment} from $u$ to $v$ can be lower than $l_{u,v}(r)$ if $u$ is in default. In this case, the model assumes that $u$ makes payments based on the \textit{principle of proportionality}, i.e. it uses all of its assets to make payments to creditors, in proportion to the respective liabilities. In practice, this means that $u$ can pay an $r_u$ portion of each liability, and thus its payment to $v$ is defined as $p_{u,v}(r) = r_u \cdot l_{u,v}(r).$

On the other hand, the \textit{assets} of $v$ is the sum of its external assets and its incoming payments, i.e.
\[ a_v(r) = e_v + \sum_{u \in V} p_{u,v}(r). \]

Recall that a recovery rate describes the portion of liabilities that a bank can pay. Hence given the assets and liabilities of each bank $v$, the recovery rate $r_v$ must satisfy $r_v=1$ if $a_v(r) \geq l_v(r)$, and $r_v=\frac{a_v(r)}{l_v(r)}$ otherwise. A vector $r$ is called a \textit{solution} (in financial terms: a clearing vector) if it describes an equilibrium point for these equalities, i.e. if for each bank $v$, $r_v$ satisfies this constraint for the assets and liabilities defined by $r$. Previous work has expressed this by defining the \textit{update function} $f\,:\,[0,1]^{V} \rightarrow [0,1]^{V}$ as
\[ f_v(r) =
\begin{cases}
    1, & \text{if } a_v(r) \geq l_v(r) \\
		\frac{a_v(r)}{l_v(r)}, & \text{if } a_v(r) < l_v(r)
\end{cases}
, \]
and defining a solution as a fixed point of the update function.

In order to model the utility function of nodes in the system, we define the \textit{payoff} (in financial terms: equity) of a bank $v$ as the amount of remaining assets after payments if a node is not in default, and 0 otherwise, i.e. $q_v(r) = \max(a_v(r)-l_v(r), 0).$
We assume that the aim of each bank is to maximize its own payoff.

Note that assets, liabilities and payoffs are always defined with regard to a certain recovery rate vector $r$. However, in order to simplify notation, we do not show $r$ explicitly when it is clear from the context, and instead we simply write e.g. $a_v$ or $q_v$.

Figure \ref{fig:example} shows an example financial system with three banks $u$, $v$ and $w$, with a consistent notation to that of \cite{base1, base2}. The system has $e_u=2$, $e_v=1$ and $e_w=0$. There are two debts of weight 2 in the system: one from $u$ to $v$, the other from $u$ to $w$. Finally, the system contains a CDS from $w$ to $v$ (also of weight 2), which is in reference to bank $u$.

Regardless of recovery rates, bank $u$ has liabilities $l_u=4$ and assets $a_u=2$, so $r_u=\frac{1}{2}$ in any case. This implies that $u$ can only make payments of $r_u \cdot 2 = 1$ to both $v$ and $w$.
Given $r_u=\frac{1}{2}$, the CDS induces a liability of $2 \cdot (1-r_u)=1$ from $w$ to $v$. Since $w$ receives an incoming payment of $p_{u,w}=1$ from $u$, we have $a_w=l_w=1$, so $w$ can still pay its liability and has a recovery rate of $r_w=1$. Finally, $v$ has incoming payments $p_{u,v}=1$ and $p_{w,v}=1$, external assets $e_w=1$, and no liabilities. This implies $a_v=3$ and $l_v=0$, and thus $r_v=1$. Hence $(r_u, r_v, r_w)=(\frac{1}{2}, 1, 1)$ is the only solution of the system, providing a payoff of $q_u=0$, $q_w=0$ and $q_v=3$ to the banks.

\begin{figure}
\centering

\begin{tikzpicture}
	
	\draw[very thick, blue, arrows=-latex] (0pt,0pt) -- (72pt,0pt);
	\draw[very thick, blue, arrows=-latex] (0pt,0pt) -- (36pt,62.1pt);
	\draw[very thick, brown, arrows=-latex] (40pt,69pt) -- (76pt,6.9pt);
	\draw[very thick, brown, densely dotted] (60pt,34.5pt) -- (0pt,0pt);
	
	\node[anchor=center] at (40pt,-7pt) {\normalsize $2$};
	\node[anchor=center] at (15pt,37pt) {\normalsize $2$};
	\node[anchor=center] at (65pt,37pt) {\normalsize $2$};
	
	\draw[black, fill=white] (0pt,0pt) circle (1.7ex);
	\draw[black, fill=white] (80pt,0pt) circle (1.7ex);
	\draw[black, fill=white] (40pt,69pt) circle (1.7ex);
	
	\node[anchor=center] at (0pt,0pt) {\normalsize $u$};
	\node[anchor=center] at (80pt,0pt) {\normalsize$v$};
	\node[anchor=center] at (40pt,69pt) {\normalsize $w$};
	
	\draw [fill=white] (4pt,-2pt) rectangle (10pt,-11pt);
	\node[anchor=center] at (7pt,-6.5pt) {\small $2$};
	\draw [fill=white] (84pt,-2pt) rectangle (90pt,-11pt);
	\node[anchor=center] at (87pt,-6.5pt) {\small $1$};
	\draw [fill=white] (44pt,67pt) rectangle (50pt,58pt);
	\node[anchor=center] at (47pt,62.5pt) {\small $0$};
	
\end{tikzpicture}
\caption{Example financial system with three banks. External assets are shown in rectangles besides the nodes, simple debt contracts are shown as blue arrows from debtor to creditor, and CDSs are shown as brown arrows from debtor to creditor, with a dotted line specifying the reference entity.}
\label{fig:example}
\end{figure}
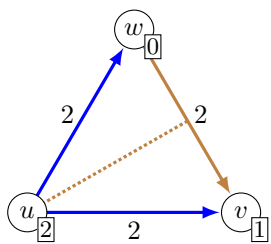

We also use two sanity assumptions introduced by previous work to exclude degenerate cases \cite{base1}. First, we assume that no bank enters into a contract with itself or in reference to itself. Furthermore, since CDSs are regarded as an insurance on debt, we require that if a bank $w$ is a reference entity of some CDS, then $w$ is the debtor of at least one debt contract of positive weight.

\section{Payments with priorities} \label{sec:prio}

While the principle of proportionality is a simple and natural assumption, financial systems often have more complex payment rules in practice. Thus we also introduce a more general model of \emph{payments with priorities}.

That is, we assume that there is a constant number of priority classes $P$, and each contract belongs to one of these priority classes. If a node $v$ is in default, then it first spends all its assets to fulfill its liabilities in the highest priority class. If $v$ does not have enough assets to fulfill all such obligation, it spends all its assets on the payments for these edges, proportionally to the amount of liabilities. On the other hand, if $v$ has more assets than highest-priority liabilities, then $v$ pays for all the liabilities in this highest priority level, and continues using the rest of its assets for the lower-priority liabilities in a similar fashion.

More formally, in our modified model, each contract in the network receives another \emph{priority} parameter (besides its weight), which is an integer in $\{ 1, ..., P \}$. The value 1 denotes the highest priority (i.e. liabilities that have to be paid first), while class $P$ denotes the lowermost priority level.

Given a clearing vector $r$, for each node $v$, let $l_v^{(\rho)}$ denote the total amount of liabilities of $v$ due to edges on priority level $\rho$. Let us also introduce the notation $ l_v^{(\leq \rho)} = \sum_{i=1}^\rho l_v^{(i)}$. Assume that $v$ has total assets of $a_v$, and a liability of $l_{v,u}$ on priority level $\rho$ towards another node $u$. Then the payment of $v$ to $u$ is defined as
\[ p_{v,u} =
\begin{cases}
    0, & \text{if } a_v \leq l_v^{(\leq \rho-1)} \\
		\frac{a_v - l_v^{(\leq \rho-1)}}{l_v^{(\rho)}} \cdot l_{v,u}, & \text{if } a_v \in \left( l_v^{(\leq \rho-1)}, l_v^{(\leq \rho)} \right) \\
    l_{v,u}, & \text{if } a_v \geq l_v^{(\leq \rho).}
\end{cases}
 \]

For an example, consider a modified version of the network in Figure \ref{fig:example}. Assume we now have 2 priority levels: the debt from $u$ to $w$ is on the higher level, while the other two contracts are on the lower level. For the case of $u$, this still means $l_u=4$, $a_u=2$ and $r_u=\frac{1}{2}$ as before. However, now $u$ uses its 2 units of assets to pay its full liability to $w$, since this contract has higher priority than the debt to $v$. Hence $p_{u,v}=0$ and $p_{u,w}=2$, resulting in $a_w=2$. Since $r_u=\frac{1}{2}$ still implies $l_w=1$ for the CDS, the rest of the payments and recovery rates remain unchanged: we still have $p_{w,v}=1$ and $r_w=r_v=1$. However, the payoffs of the banks in the system are now $q_u=0$, $q_w=1$ and $q_v=2$.

The main motivation for introducing payment priorities is that in many cases, it is very close to what happens in real-world financial systems. In many countries, economic laws provide a specific priority list for companies to follow when paying their debts in case of a default. This might start with salaries and other payments to the employees of the company first, then specific kind of debt contracts, and so on. 

Another advantage of priorities is that we can use them to replace so-called \emph{default costs}. Default costs (also studied in \cite{base1, base2}) are an extension of the original model, assuming that when banks go into default, they immediately lose a specific portion of their assets. This represents the fact that in practice, once a company goes into default, it has a range of immediate payment obligations (e.g. employees' wages) before it can make payments to other banks in the system. If we instead represent the bank's employees as a separate node in the network, and model this payment obligation with a high-priority edge, then this allows us to describe the phenomenon \textit{without} the use of default costs.

This observation is crucial because the introduction of default costs comes at a significant price: intuitively speaking, default costs introduce a point of discontinuity into the update function, and as a result, some financial systems do not have a solution at all \cite{base1}. In contrast to this, without default costs, systems always have at least one solution, as shown by a fixed-point argument in \cite{base1}. We point out that the same fixed-point theorem proof also applies in our model with payment priorities: even though the functions $p_{u,v}(r)$ and $a_v(r)$ become significantly more complicated, they are still continuous.

This shows that by introducing priorities, we obtain a model that is significantly more realistic on one hand, but also ensures the existence of a solution at the same time.

\begin{theorem}
Every financial system with payment priorities has at least one solution.
\end{theorem}

\renewcommand{\proofname}{Proof (sketch).}

\begin{proof}
The proof of this claim is identical to the same proof in the original financial system model, described in the results of \cite{base1}. The main idea of the proof is to apply the fixed-point theorem of Kakutani \cite{kakutani}, which ensures the existence of a fixed point of the update function $f$, and thus a solution. This proof can still be applied after the introduction of priorities, since both $a_v(r)$ and $l_v(r)$ still remains a continuous function of $r$, and so does the update function $f_v(r)=\min(\frac{a_v(r)}{l_v(r)}, 1)$, at least in the domain where $l_v(r)>0$. The technical part of the proof is slightly more complicated, since one has to consider the $l_v(r)=0$ case separately. For more details on this proof, we refer the reader to the work of \cite{base1}.
\end{proof}

\renewcommand{\proofname}{Proof.}

\section{Influencing the financial system}

We now discuss a wide range of actions that a bank can execute in order to obtain a more favorable outcome in the system. Note that except for Section \ref{sec:reprio} which explicitly studies readjusting priorities, all the results also hold in the base model without priorities.

\subsection{Removing an incoming debt}

One of the most natural actions for a bank $v$ would be to simply cancel a debt contract in which $v$ is a creditor. Since the creditor is considered the beneficiary of a debt, in some financial/legal frameworks, the regulations may indeed allow a bank to nullify an incoming debt contract. However, in case of a financial system with short positions, it is actually possible that in the end, this indirectly increases the payoff of $v$.

\begin{theorem} \label{th:removedebt}
Removing an incoming debt of $v$ can increase the payoff of $v$.
\end{theorem}

More precisely, our claim is as follows: there exists a financial system $S$ such that (i) $S$ has only one solution $r$, in which $v$ has payoff $q_v$ and an incoming debt contract, and (ii) in the modified financial system $S'$ obtained by removing this debt, there is again only one solution $r'$, in which the payoff $q_v'$ satisfies $q_v' > q_v$.

\begin{proof}
Consider the network in Figure \ref{fig:remove}. Note that the unlabeled nodes in this system can always pay all their liabilities, so their recovery rate is always 1. Originally, the system has $a_u=1$ and $l_u=2$, thus $r_u=\frac{1}{2}$ in any case. This implies $a_w=2 \cdot (1-\frac{1}{2}) = 1$, and thus $r_w=1$. With $r_w=1$, $v$ obtains no payment from its incoming CDS at all, so the payoff of $v$ in this only solution is $q_v=p_{u,v}=r_u \cdot 1 = \frac{1}{2}$.

One the other hand, consider the system obtained by removing the debt contract from $u$ to $v$. In this case, $a_u=l_u=1$, and thus $r_u=1$. This means that $w$ receives no incoming payments at all, and with $a_w=0$, we have $r_w=0$. As a result, $v$ obtains a payment of $2 \cdot (1-r_w)=2$ from its incoming CDS, so we have $q_v=2$.
\end{proof}

\begin{figure}
\centering
\captionsetup{justification=centering}
\hspace{0.02\textwidth}
\minipage{0.4\textwidth}
	\centering

\begin{tikzpicture}
	
	\draw[very thick, blue, arrows=-latex] (30pt,80pt) -- (3pt,35pt);
	\draw[very thick, blue, arrows=-latex] (30pt,80pt) -- (84pt,62pt);
	\draw[very thick, blue, arrows=-latex] (60pt,30pt) -- (86pt,4pt);
	\draw[very thick, brown, arrows=-latex] (0pt,30pt) -- (54pt,30pt);
	\draw[very thick, brown, arrows=-latex] (90pt,0pt) -- (90pt,54pt);
	\draw[very thick, brown, densely dotted] (90pt,30pt) -- (60pt,30pt);
	\draw[very thick, brown, densely dotted] (30pt,30pt) -- (30pt,80pt);
	
	\node[anchor=center] at (30pt,23pt) {\normalsize $2$};
	\node[anchor=center] at (96pt,30pt) {\normalsize $2$};
	\node[anchor=center] at (65pt,73pt) {\normalsize $1$};
	\node[anchor=center] at (12pt,59pt) {\normalsize $1$};
	\node[anchor=center] at (72pt,10pt) {\normalsize $1$};
	
	\draw[black, fill=white] (0pt,30pt) circle (1.7ex);
	\draw[black, fill=white] (60pt,30pt) circle (1.7ex);
	\draw[black, fill=white] (30pt,80pt) circle (1.7ex);
	\draw[black, fill=white] (90pt,60pt) circle (1.7ex);
	\draw[black, fill=white] (90pt,0pt) circle (1.7ex);
	
	\node[anchor=center] at (60pt,30pt) {\normalsize $w$};
	\node[anchor=center] at (30pt,80pt) {\normalsize $u$};
	\node[anchor=center] at (90pt,60pt) {\normalsize $v$};
	
	\draw [fill=white] (4pt,28pt) rectangle (10pt,19pt);
	\node[anchor=center] at (7pt,23.5pt) {\small $2$};
	\draw [fill=white] (64pt,28pt) rectangle (70pt,19pt);
	\node[anchor=center] at (67pt,23.5pt) {\small $0$};
	\draw [fill=white] (34pt,78pt) rectangle (40pt,69pt);
	\node[anchor=center] at (37pt,73.5pt) {\small $1$};
	\draw [fill=white] (94pt,58pt) rectangle (100pt,49pt);
	\node[anchor=center] at (97pt,53.5pt) {\small $0$};
	\draw [fill=white] (94pt,-2pt) rectangle (100pt,-11pt);
	\node[anchor=center] at (97pt,-6.5pt) {\small $2$};
	
\end{tikzpicture}
	\caption{Example for removing an incoming debt}
	\label{fig:remove}
\endminipage\hfill
\hspace{0.09\textwidth}
\minipage{0.4\textwidth}
	\centering

\begin{tikzpicture}
	
	\draw[very thick, blue, arrows=-latex] (30pt,80pt) -- (3pt,35pt);
	\draw[very thick, blue, arrows=-latex] (30pt,80pt) -- (84pt,62pt);
	\draw[very thick, blue, arrows=-latex] (60pt,30pt) -- (86pt,4pt);
	\draw[very thick, brown, arrows=-latex] (0pt,30pt) -- (54pt,30pt);
	\draw[very thick, brown, arrows=-latex] (90pt,0pt) -- (90pt,54pt);
	\draw[very thick, brown, densely dotted] (90pt,30pt) -- (60pt,30pt);
	\draw[very thick, brown, densely dotted] (30pt,30pt) -- (30pt,80pt);
	
	\node[anchor=center] at (30pt,24pt) {\footnotesize $2/\gamma_0$};
	\node[anchor=center] at (96pt,30pt) {\normalsize $2$};
	\node[anchor=center] at (65pt,73pt) {\normalsize $1$};
	\node[anchor=center] at (12pt,59pt) {\normalsize $1$};
	\node[anchor=center] at (72pt,10pt) {\normalsize $1$};
	
	\draw[black, fill=white] (0pt,30pt) circle (1.7ex);
	\draw[black, fill=white] (60pt,30pt) circle (1.7ex);
	\draw[black, fill=white] (30pt,80pt) circle (1.7ex);
	\draw[black, fill=white] (90pt,60pt) circle (1.7ex);
	\draw[black, fill=white] (90pt,0pt) circle (1.7ex);
	
	\node[anchor=center] at (60pt,30pt) {\normalsize $w$};
	\node[anchor=center] at (30pt,80pt) {\normalsize $u$};
	\node[anchor=center] at (90pt,60pt) {\normalsize $v$};
	
	\draw [fill=white] (-2pt,27.5pt) rectangle (13.5pt,19.7pt);
	\node[anchor=center] at (6pt,23.5pt) {\scriptsize $2/\!\gamma_0$};
	\draw [fill=white] (64pt,28pt) rectangle (70pt,19pt);
	\node[anchor=center] at (67pt,23.5pt) {\small $0$};
	\draw [fill=white] (28.5pt,76.5pt) rectangle (45pt,69pt);
	\node[anchor=center] at (37pt,72.5pt) {\scriptsize $2\!\!-\!\!\gamma_0$};
	\draw [fill=white] (94pt,58pt) rectangle (100pt,49pt);
	\node[anchor=center] at (97pt,53.5pt) {\small $0$};
	\draw [fill=white] (94pt,-2pt) rectangle (100pt,-11pt);
	\node[anchor=center] at (97pt,-6.5pt) {\small $2$};
	
\end{tikzpicture}
	\caption{Removing $\gamma_0$ portion of an incoming debt}
	\label{fig:partial}
\endminipage\hfill
\hspace{0.01\textwidth}
\end{figure}

The proof shows that releasing an outgoing debt increases the recovery rate of $u$, which indirectly yields an extra payoff for $v$. Note that $v$ could also achieve this result by donating funds to $u$, i.e. by increasing $e_u$ by $1$. This is even more realistic in a legal framework: the owner(s) of bank $v$ can simply donate a specific amount to bank $u$, who would accept it in hope of avoiding default. Naturally, this is only a favorable step to $v$ if by donating $x$ units of money, it can increase its own payoff by more than $x$.

\begin{theorem} \label{th:injectu}
Donating external assets to another node $u$ can be a favorable step. 
\end{theorem}

More precisely, there is a system $S$ such that (i) $S$ has only one solution $r$, in which node $v$ has payoff $q_v$, and (ii) in the system $S'$ obtained by replacing external funds of $u$ by $e_u':=e_u+x$, there is again only one solution $r'$ which satisfies $q'_v > q_v+x$.

The proof of the theorem is identical to that of Theorem \ref{th:removedebt}: if $v$ increases $e_u$ by $x=1$ in Figure \ref{fig:remove}, then again $r_u=1$, which ultimately provides a payoff of $q_v=3$ (as opposed to the original $\frac{1}{2}$). Note that in general, this action may allow banks to improve their position by affecting a bank that is arbitrarily far in the topology of the network.

Finally, if $v$ can increase its payoff by releasing an incoming debt, it is natural to wonder if it is always optimal for $v$ to erase the entire debt, or whether it could be beneficial to only reduce the amount in some cases. We show that reducing a debt to a given portion $\gamma_0$ of its original weight can also be an optimal strategy.

\begin{theorem}
For each constant $\gamma_0 \in [0,1]$, there is a financial system where bank $v$ achieves its maximal payoff by reducing an incoming debt to a $\gamma_0$ portion of its original weight. 
\end{theorem}

\begin{proof}
Consider a modified version of our previous systems, as shown in Figure \ref{fig:partial}. We show that for any $\gamma_0$ parameter, the optimal action of $v$ in this system is to let go of $\gamma_0$ portion of the incoming debt from $u$, i.e. to reduce its weight to $1-\gamma_0$.

Assume that $v$ reduces the incoming debt by a $\gamma$ portion for some $\gamma \in [0,1]$, and let us analyze the final payoff of $v$ as a function of $\gamma$. Note that a choice of $\gamma=\gamma_0$ implies that $a_u=l_u$ exactly, and thus $r_u=1$, $r_w=0$ and $q_v=(1-\gamma_0)+2=3-\gamma_0$ as a result. Hence we have to show that $q_v<3-\gamma_0$ in any other case.

First consider the case when $\gamma < \gamma_0$. Since $u$ has $l_u=1+(1-\gamma)=2-\gamma > 2-\gamma_0$, $u$ is in default. Then $r_u=\frac{2-\gamma_0}{2-\gamma}$, and thus $w$ receives an incoming payment of
\[ a_w = \frac{2}{\gamma_0} \cdot \left( 1 - \frac{2-\gamma_0}{2-\gamma} \right) = \frac{2 \cdot (\gamma_0-\gamma)}{\gamma_0 \cdot (2-\gamma)}. \]
This is a decreasing function in $\gamma$, and it equals $1$ exactly for $\gamma=0$, so $a_w<1$ for any $\gamma>0$, and thus $w$ is in default with $r_w=a_w$. Then the amount $v$ receives from the CDS is
\[ 2 \cdot (1-r_w) = 2 \cdot \left( 1 - \frac{2 \cdot (\gamma_0-\gamma)}{\gamma_0 \cdot (2-\gamma)} \right) = 2 \cdot \frac{\gamma \cdot (2 - \gamma_0)}{\gamma_0 \cdot (2 - \gamma)}. \]
Since $q_v=(1-\gamma) \cdot r_u + 2 \cdot (1-r_w)$, we need to show that
\[ 3-\gamma_0 > (1-\gamma) \cdot \frac{2-\gamma_0}{2-\gamma} + 2 \cdot \frac{\gamma \cdot (2 - \gamma_0)}{\gamma_0 \cdot (2 - \gamma)}. \]
After multiplying this by $\gamma_0 \cdot (2-\gamma)$, expanding the brackets and removing terms that cancel out, we are left with $\gamma_0 \cdot (4-\gamma_0) > \gamma \cdot (4-\gamma_0)$, which naturally holds since $\gamma < \gamma_0$.

On the other hand, if $\gamma > \gamma_0$, then $a_u > l_u$, and thus $r_u=1$. This means $r_w=0$, so $v$ receives an amount of $2$ from the CDS, and has a total payoff of $(1-\gamma) + 2 = 3- \gamma$, which is again less than $3-\gamma_0$.
Thus selecting $\gamma=\gamma_0$ is indeed the best option for $v$.
\end{proof}

\subsection{Investing more external assets}

In light of Theorem \ref{th:injectu}, it is natural to ask if $v$ can also increase its payoff by injecting further funds into its own external assets. That is, if increasing $e_v$ by $x$ would allow $v$ to increase its payoff by more than $x$ in the only solution, then the owner of bank $v$ would be motivated to invest these extra funds into the bank.

However, somewhat surprisingly, it turns out that this is not possible in the same way as in previous cases: we cannot increase the payoff of $v$ by more than $x$ in the only solution of the system. More specifically, if a vector $r'$ is a solution to the new system 
and provides a payoff of $q_v'$, then $r'$ was already a solution of the original system with a payoff of $q_v'-x$.

\begin{theorem}
Assume that every solution of system $S$ provides a payoff of at most $q_v$ for $v$. Then 
setting $e_v'=e_v+x$ cannot introduce a new solution $r'$ with $q_v' > q_v+x$.
\end{theorem}

\begin{proof}
Assume that such a new solution $r'$ is introduced. Since payoff is always nonnegative, $q_v \geq 0$, and thus $q_v' > x$ in $r'$. This means that we have $a_v' > x + l_v'$ in $r'$. Hence, even if $e_v'$ was reduced by $x$ (back to its original value $e_v$), then $v$ could still pay all of its liabilities; thus the same recovery vector $r'$ and the same payments on each edge also provide a solution in the original system $S$. The payoff of $v$ in this solution is $q_v'-x$, which is larger than $q_v$ by assumption. This contradicts the fact that $q_v$ was the maximal payoff for $v$ in $S$.
\end{proof}

Naturally, if $v$ is in default, then recovery rate of $v$ can indeed be increased in the only solution by injecting extra funds. However, an increase of $r_v$ does not translate to an increase in payoff, so it is a waste for the owners of $v$ to invest resources for this.

On the other hand, while its not possible to produce a new, more favorable solution for $v$, it is possible to invalidate solutions that are unfavorable to $v$. That is, if the original financial system had multiple solutions with different payoffs for $v$, and $v$ is unsure which of these solutions will be implemented by a financial authority, then it is possible that $v$ can inject extra funds to remove a solution where its payoff is much smaller than in other solutions. This may allow $v$ to increase its worst-case payoff, or its payoff in expectation (in case of a randomized choice of solution).

\begin{theorem}
Given a financial system $S$ with two solutions, it is possible that setting $e_v'=e_v+x$ removes the solution which is unfavorable to $v$.
\end{theorem}

More precisely, there is a system $S$ such that (i) $S$ has two solutions $r_1$ and $r_2$, with solution $r_2$ satisfying $q_v=0$, and (ii) in the system $S'$ obtained by setting $e_v':=e_v+x$, the only solution is $r'=r_1$, satisfying $q_v' > x$.

\begin{proof}
Consider the system in Figure \ref{fig:inject}, which has two solutions. The design of the system ensures $r_u=r_v$ and $r_w=1-r_u$. If $r_v=1$, then this implies $r_u=1$ and $r_w=0$, in which case $v$ has $a_v=100$, giving a solution with $q_v=99$. On the other hand, if $r_v<1$, then it has to satisfy
\[ r_v=\frac{100 \cdot (1-r_w)}{1} = 100 \cdot r_u = 100 \cdot r_v. \]
This is only satisfied if $r_v=0$, so this is the only other solution, providing $q_v=0$.

Now assume that $v$ invests $x=1$ extra funds to have $e_v=1$. In this case, the system always has $r_v=1$, hence $r_u=1$ and $r_w=0$. This implies that $v$ obtains a payment of $100$ in the CDS, resulting in a payoff of $q_v=100$. Even if we subtract the extra $x=1$ investment, $v$ has an extra payoff of $99$, and thus it has indeed increased its worst-case payoff significantly.
\end{proof}

\begin{figure}
\centering
\captionsetup{justification=centering}
\hspace{0.02\textwidth}
\minipage{0.4\textwidth}
	\centering

\begin{tikzpicture}
	
	\draw[very thick, blue, arrows=-latex] (30pt,80pt) -- (3pt,35pt);
	\draw[very thick, blue, arrows=-latex] (90pt,80pt) -- (36pt,80pt);
	\draw[very thick, blue, arrows=-latex] (60pt,30pt) -- (86pt,4pt);
	\draw[very thick, brown, arrows=-latex] (0pt,30pt) -- (54pt,30pt);
	\draw[very thick, brown, arrows=-latex] (90pt,0pt) -- (90pt,74pt);
	\draw[very thick, brown, densely dotted] (90pt,40pt) -- (60pt,30pt);
	\draw[very thick, brown, densely dotted] (30pt,30pt) -- (30pt,80pt);
	
	\node[anchor=center] at (30pt,23pt) {\normalsize $1$};
	\node[anchor=center] at (99pt,40pt) {\normalsize $100$};
	\node[anchor=center] at (65pt,73pt) {\normalsize $1$};
	\node[anchor=center] at (12pt,59pt) {\normalsize $1$};
	\node[anchor=center] at (72pt,10pt) {\normalsize $1$};
	
	\draw[black, fill=white] (0pt,30pt) circle (1.7ex);
	\draw[black, fill=white] (60pt,30pt) circle (1.7ex);
	\draw[black, fill=white] (30pt,80pt) circle (1.7ex);
	\draw[black, fill=white] (90pt,80pt) circle (1.7ex);
	\draw[black, fill=white] (90pt,0pt) circle (1.7ex);
	
	\node[anchor=center] at (60pt,30pt) {\normalsize $w$};
	\node[anchor=center] at (30pt,80pt) {\normalsize $u$};
	\node[anchor=center] at (90pt,80pt) {\normalsize $v$};
	
	\draw [fill=white] (4pt,28pt) rectangle (10pt,19pt);
	\node[anchor=center] at (7pt,23.5pt) {\small $1$};
	\draw [fill=white] (64pt,28pt) rectangle (70pt,19pt);
	\node[anchor=center] at (67pt,23.5pt) {\small $0$};
	\draw [fill=white] (34pt,78pt) rectangle (40pt,69pt);
	\node[anchor=center] at (37pt,73.5pt) {\small $0$};
	\draw [fill=white] (94pt,78pt) rectangle (100pt,69pt);
	\node[anchor=center] at (97pt,73.5pt) {\small $0$};
	\draw [fill=white] (90.2pt,-2.5pt) rectangle (104pt,-10.5pt);
	\node[anchor=center] at (97pt,-6.5pt) {\footnotesize $100$};
	
\end{tikzpicture}
	\caption{A bank $v$ increasing its own external assets}
	\label{fig:inject}
\endminipage\hfill
\hspace{0.09\textwidth}
\minipage{0.4\textwidth}	
  \centering

\begin{tikzpicture}
	
	\draw[very thick, blue, arrows=-latex] (30pt,80pt) -- (3pt,35pt);
	\draw[very thick, blue, arrows=-latex] (90pt,80pt) -- (36pt,80pt);
	\draw[very thick, blue, arrows=-latex] (90pt,80pt) -- (116pt,54pt);
	\draw[very thick, blue, arrows=-latex] (60pt,30pt) -- (86pt,4pt);
	\draw[very thick, brown, arrows=-latex] (0pt,30pt) -- (54pt,30pt);
	\draw[very thick, brown, arrows=-latex] (90pt,0pt) -- (90pt,74pt);
	\draw[very thick, brown, densely dotted] (90pt,40pt) -- (60pt,30pt);
	\draw[very thick, brown, densely dotted] (30pt,30pt) -- (30pt,80pt);
	
	\node[anchor=center] at (30pt,23pt) {\normalsize $2$};
	\node[anchor=center] at (96pt,40pt) {\normalsize $\delta$};
	\node[anchor=center] at (65pt,73pt) {\normalsize $1$};
	\node[anchor=center] at (108pt,70pt) {\normalsize $1$};
	\node[anchor=center] at (12pt,59pt) {\normalsize $1$};
	\node[anchor=center] at (72pt,10pt) {\normalsize $1$};
	
	\draw[black, fill=white] (0pt,30pt) circle (1.7ex);
	\draw[black, fill=white] (60pt,30pt) circle (1.7ex);
	\draw[black, fill=white] (30pt,80pt) circle (1.7ex);
	\draw[black, fill=white] (90pt,80pt) circle (1.7ex);
	\draw[black, fill=white] (90pt,0pt) circle (1.7ex);
	\draw[black, fill=white] (120pt,50pt) circle (1.7ex);
	
	\node[anchor=center] at (60pt,30pt) {\normalsize $w$};
	\node[anchor=center] at (30pt,80pt) {\normalsize $u$};
	\node[anchor=center] at (90pt,80pt) {\normalsize $v$};
	
	\draw [fill=white] (4pt,28pt) rectangle (10pt,19pt);
	\node[anchor=center] at (7pt,23.5pt) {\small $2$};
	\draw [fill=white] (64pt,28pt) rectangle (70pt,19pt);
	\node[anchor=center] at (67pt,23.5pt) {\small $0$};
	\draw [fill=white] (34pt,78pt) rectangle (40pt,69pt);
	\node[anchor=center] at (37pt,73.5pt) {\small $0$};
	\draw [fill=white] (124pt,48pt) rectangle (130pt,39pt);
	\node[anchor=center] at (127pt,43.5pt) {\small $0$};
	\draw [fill=white] (94pt,78pt) rectangle (100pt,69pt);
	\node[anchor=center] at (97pt,73.5pt) {\small $1$};
	\draw [fill=white] (94pt,-2pt) rectangle (100pt,-11pt);
	\node[anchor=center] at (97pt,-6.5pt) {\footnotesize $\delta$};
	
\end{tikzpicture}
	\caption{Readjusting the priority of outgoing contracts}
	\label{fig:reprio}
\endminipage\hfill
\hspace{0.01\textwidth}
\end{figure}

\subsection{Readjusting priorities} \label{sec:reprio}

Assuming payments with priorities as discussed in Section \ref{sec:prio}, it is also interesting to know if a node can improve its situation by readjusting the priorities of its outgoing edges. That is, in a more flexible regulation framework, banks may be allowed to choose to some extent the order in which they fulfill their payment obligations. 
However, we show that similarly to the previous case, readjusting the priorities of outgoing edges cannot introduce a better solution.

\begin{theorem}
Assume that every solution of system $S$ provides a payoff of at most $q_v$ for $v$. Then redefining $v$'s outgoing priorities cannot introduce a new solution $r'$ with $q_v' > q_v$.
\end{theorem}

\begin{proof}
Assume that such a new solution $r'$ is introduced. Payoff is nonnegative, so $q_v \geq 0$, and thus $q_v' > 0$. This implies that $a_v' > l_v'$ in $r'$, i.e. $v$ is able to pay all of its liabilities in every outgoing contract. However, in this case, the priorities on the outgoing edges do not matter; hence $r'$ is a solution of $S'$ regardless of how the priorities of outgoing contracts are chosen. In particular, $r'$ is already a solution of the initial system $S$ before the priorities were reorganized, giving the same payoff $q_v'$ in $S$. This contradicts the fact that $q_v$ was the maximal payoff for $v$ in $S$.
\end{proof}

However, it is again possible that $v$ can increase its recovery rate by readjusting priorities. Recall that in the previous case of increasing the bank's own external assets, we did not explore this possibility, since it required the bank $v$ to invest extra funds while not yielding (the same amount of) extra payoff. However, readjusting priorities is an action that $v$ might be able to execute free of charge. Thus if we define the recovery rate as the secondary objective function of a bank (i.e. even if $v$ is in default and thus has 0 payoff, it is not oblivious to the outcome, and prefers a higher recovery rate), then redefining priorities may allow $v$ to achieve a more preferred outcome without having to invest any extra funds.

\begin{theorem} \label{th:prio_rec}
Redefining $v$'s outgoing priorities can increase the recovery rate of $v$.
\end{theorem}

\begin{proof}
Consider the system in Figure \ref{fig:reprio} with a choice of $\delta=\frac{1}{2}$. Originally, each contract is in the same (lower) priority class. Bank $v$ never has enough assets to pay its liabilities, hence $u$ is also in default. In this case, we have $r_u=r_v$ and $r_w=2-2 \cdot r_u$, so $v$ receives $\delta \cdot (1-r_w)=\delta \cdot (2 \cdot r_v-1)$ funds from the CDS. This means that 
\[ r_v=\frac{\delta \cdot (2 \cdot r_v-1)+1}{2}, \]
which, after reorganization, gives $\delta-1 = 2 \cdot (\delta-1) \cdot r_v$, and thus $r_v=\frac{1}{2}$. This is the only solution of the system if $\delta \neq 1$.

Now assume that $v$ is able to raise the debt towards $u$ to the higher priority level. In this new system, $v$ first fulfills its payment obligation to $u$, which is always possible from its external assets. Hence $r_u=1$ in this case, implying $r_w=0$ and thus a payment of $\frac{1}{2}$ to $v$ in the CDS. This implies $r_v=\frac{3}{4}$ in the only solution of the new system.
\end{proof}

Finally, we show that redefining priorities can allow $v$ to remove an unfavorable solution, and thus increase its worst-case or expected payoff as in the previous subsection.

\begin{theorem}
Given a financial system $S$ with two solutions, redefining $v$'s outgoing priorities can remove the solution which is unfavorable to $v$.
\end{theorem}

\begin{proof}
Consider the system in Figure \ref{fig:reprio} with a choice of $\delta=100$. As discussed in the proof of Theorem \ref{th:prio_rec}, if $r_v<1$, then the only solution is $r_v=\frac{1}{2}$. However, the large $\delta$ value now allows another solution in the original system: if $r_v=1$, then $r_u=1$ and $r_w=0$, ensuring that $v$ indeed has enough funds to pay its liabilities. The two solutions come with payoffs of $q_v=0$ and $q_v=98$, respectively.

Now if $v$ raises its debt towards $u$ to the higher priority level, then $r_u=1$ is always guaranteed, so $r_w=0$ and thus $v$ indeed has a payoff of $98$ in the only solution.
\end{proof}

\section{Game-theoretic dilemmas in financial systems}

Finally, we briefly show that the attempts of banks to influence the system can also easily lead to situations that can be described by classical game-theoretic settings.

\subsection{Prisoner's dilemma}

We first show an example where if two nodes simultaneously try to influence the system to their advantage, then the resulting situation is essentially identical to the well-known prisoner's dilemma \cite{generalGT}.

Consider the financial system in Figure \ref{fig:prisoners}, where banks $v_1$ and $v_2$ want to influence the system to achieve a better outcome. Assume that in the current legal framework, the only step available to these banks is to completely remove their incoming debt contract from $u$ (as in Theorem \ref{th:removedebt}); both banks can decide whether to execute this step or not. Note that canceling a debt increases the recovery rate of $u$, which indirectly implies a larger payment on the CDS for both $v_1$ and $v_2$, and thus can be beneficial for both banks. Applying prisoner's dilemma terminology, we also refer to the step of canceling the debt as \textit{cooperation}, and the step of not canceling the debt as \textit{defection}.

Now let us analyze the payoff of $v_1$ and $v_2$ in each strategy profile. Note that $r_w=1-r_u$, so the payment on the CDSs for both $v_1$ and $v_2$ is $3 \cdot (1-r_w) = 3 \cdot r_u$ in any case.

If both of the nodes cooperate (i.e. both debts are removed), then $u$ can pay its remaining liabilities, thus $r_u=1$. This implies a payment of $3$ on the CDS, which is the only asset of the acting nodes in this case; hence $q_{v_1}=q_{v_2}=3$.

If both of the nodes defect (no debt is removed), then we only have $r_u=\frac{1}{3}$, resulting in a payment of $1$ from the CDS. However, in this case, both $v_1$ and $v_2$ also get a direct payment of $5 \cdot r_u = \frac{5}{3}$ from the defaulting $u$, which adds up to a total payoff of $\frac{8}{3}=2.\dot{6}$.

Finally, assume that only one of the nodes cooperate (say, $v_1$). With only one of the outgoing debts removed, $u$ will have a recovery rate of $r_u=\frac{1}{2}$. This results in a payment of $\frac{3}{2}$ on the CDS for both nodes. However, note that $v_2$ still has an incoming debt contract from $u$, and receives a payment of $5 \cdot r_u = \frac{5}{2}$ on this contract. This implies $q_{v_1}=\frac{3}{2}=1.5$, while $q_{v_2}=4$ for the strategy profile. The symmetric case yields $q_{v_1}=4$ and $q_{v_2}=1.5$.

Since the four payoffs are ordered exactly as in case of a prisoner's dilemma, we obtain an essentially equivalent situation if the two banks are not allowed to coordinate. For both players, defection is always a dominant strategy. E.g. for $v_2$, defection yields a payoff of $4$ (instead of only $3$) if $v_1$ cooperates, and it yields a payoff of $2.\dot{6}$ (instead of only $1.5$) if $v_1$ defects. Thus the Nash-Equilibrium of the game is obtained when both players defect, with $q_{v_1}=q_{v_2}=2.\dot{6}$. However, both players would be better off in the social optimum of mutual cooperation, which gives $q_{v_1}=q_{v_2}=3$.

Some other well-known two-player games, e.g. the chicken or stag hunt game \cite{generalGT} can also easily occur in financial networks in a similar fashion. We outline some example financial systems that correspond to these further games in Appendix \ref{App:games}.

\begin{figure}
\centering
\captionsetup{justification=centering}
\hspace{0.02\textwidth}
\minipage{0.4\textwidth}
	\centering

\begin{tikzpicture}
	
	\draw[very thick, blue, arrows=-latex] (30pt,80pt) -- (3pt,35pt);
	\draw[very thick, blue, arrows=-latex] (30pt,80pt) -- (84pt,62pt);
	\draw[very thick, blue, arrows=-latex] (30pt,80pt) -- (84pt,98pt);
	\draw[very thick, blue, arrows=-latex] (60pt,30pt) -- (86pt,4pt);
	\draw[very thick, brown, arrows=-latex] (0pt,30pt) -- (54pt,30pt);
	\draw[very thick, brown, arrows=-latex] (90pt,0pt) -- (90pt,54pt);
	\draw[very thick, brown, arrows=-latex] (90pt,0pt) -- (105pt,15pt) -- (105pt,60pt) -- (90pt, 75pt) -- (90pt, 94pt);
	\draw[very thick, brown, densely dotted] (90pt,29pt) -- (60pt,29pt);
	\draw[very thick, brown, densely dotted] (105pt,33pt) -- (60pt,33pt);
	\draw[very thick, brown, densely dotted] (30pt,30pt) -- (30pt,80pt);
	
	\node[anchor=center] at (30pt,23pt) {\normalsize $1$};
	\node[anchor=center] at (95pt,27pt) {\normalsize $3$};
	\node[anchor=center] at (110pt,33pt) {\normalsize $3$};
	\node[anchor=center] at (65pt,74pt) {\normalsize $5$};
	\node[anchor=center] at (65pt,98pt) {\normalsize $5$};
	\node[anchor=center] at (11.5pt,59pt) {\normalsize $5$};
	\node[anchor=center] at (72pt,10pt) {\normalsize $1$};
	
	\draw[black, fill=white] (0pt,30pt) circle (1.7ex);
	\draw[black, fill=white] (60pt,30pt) circle (1.7ex);
	\draw[black, fill=white] (30pt,80pt) circle (1.7ex);
	\draw[black, fill=white] (90pt,60pt) circle (1.7ex);
	\draw[black, fill=white] (90pt,100pt) circle (1.7ex);
	\draw[black, fill=white] (90pt,0pt) circle (1.7ex);
	
	\node[anchor=center] at (60pt,30pt) {\normalsize $w$};
	\node[anchor=center] at (30pt,80pt) {\normalsize $u$};
	\node[anchor=center] at (90.5pt,59.5pt) {\normalsize $v_2$};
	\node[anchor=center] at (90.5pt,99.5pt) {\normalsize $v_1$};
	
	\draw [fill=white] (4pt,28pt) rectangle (10pt,19pt);
	\node[anchor=center] at (7pt,23.5pt) {\small $1$};
	\draw [fill=white] (64pt,28pt) rectangle (70pt,19pt);
	\node[anchor=center] at (67pt,23.5pt) {\small $0$};
	\draw [fill=white] (34pt,78pt) rectangle (40pt,69pt);
	\node[anchor=center] at (37pt,73.5pt) {\small $5$};
	\draw [fill=white] (94.5pt,57.5pt) rectangle (100.5pt,48.5pt);
	\node[anchor=center] at (97.5pt,53pt) {\small $0$};
	\draw [fill=white] (94.5pt,97.5pt) rectangle (100.5pt,88.5pt);
	\node[anchor=center] at (97.5pt,93pt) {\small $0$};
	\draw [fill=white] (94pt,-2pt) rectangle (100pt,-11pt);
	\node[anchor=center] at (97pt,-6.5pt) {\small $6$};
	
\end{tikzpicture}
	\vspace{-5pt}
	\caption{Prisoner's dilemma in financial systems}
	\label{fig:prisoners}
\endminipage\hfill
\hspace{0.09\textwidth}
\minipage{0.4\textwidth}	
  \centering
	\vspace{19pt}

\begin{tikzpicture}
	
	\draw[very thick, blue, arrows=-latex] (40pt,15pt) -- (75.5pt,3pt);
	\draw[very thick, blue, arrows=-latex] (40pt,-15pt) -- (75.5pt,-3pt);
	\draw[very thick, brown, arrows=-latex] (0pt,0pt) -- (35pt,14pt);
	\draw[very thick, brown, arrows=-latex] (0pt,0pt) -- (35pt,-14pt);
	\draw[very thick, brown, arrows=-latex] (0pt,0pt) -- (0pt,35pt) -- (75pt,35pt);
	\draw[very thick, brown, arrows=-latex] (0pt,0pt) -- (0pt,-35pt) -- (75pt,-35pt);
	\draw[very thick, brown, densely dotted] (40pt,15pt) -- (20pt,-7.5pt);
	\draw[very thick, brown, densely dotted] (40pt,-15pt) -- (20pt,7.5pt);
	\draw[very thick, brown, densely dotted] (40pt,15pt) -- (40pt,35pt);
	\draw[very thick, brown, densely dotted] (40pt,-15pt) -- (40pt,-35pt);
	
	\node[anchor=center] at (17pt,12pt) {\normalsize $1$};
	\node[anchor=center] at (17pt,-12pt) {\normalsize $1$};
	\node[anchor=center] at (67pt,12pt) {\normalsize $1$};
	\node[anchor=center] at (67pt,-12pt) {\normalsize $1$};
	
	\node[anchor=center] at (25pt,-29pt) {\normalsize $\delta$};
	\node[anchor=center] at (25pt,29pt) {\normalsize $\delta$};
	
	\draw[black, fill=white] (0pt,0pt) circle (1.7ex);
	\draw[black, fill=white] (40pt,-15pt) circle (1.7ex);
	\draw[black, fill=white] (40pt,15pt) circle (1.7ex);
	\draw[black, fill=white] (80pt,0pt) circle (1.7ex);
	\draw[black, fill=white] (80pt,-35pt) circle (1.7ex);
	\draw[black, fill=white] (80pt,35pt) circle (1.7ex);
	
	\node[anchor=center] at (40pt,-15pt) {\normalsize $v$};
	\node[anchor=center] at (40pt,15pt) {\normalsize $u$};
	\node[anchor=center] at (80.5pt,35.5pt) {\normalsize $u'$};
	\node[anchor=center] at (80.5pt,-34.5pt) {\normalsize $v'$};
	
	\draw [fill=white] (-3pt,-4pt) rectangle (14pt,-11pt);
	\node[anchor=center] at (5.5pt,-7.5pt) {\scriptsize $2\!\!+\!\!2 \delta$};
	\draw [fill=white] (44pt,13pt) rectangle (50pt,4pt);
	\node[anchor=center] at (47pt,8.5pt) {\small $0$};
	\draw [fill=white] (44pt,-17pt) rectangle (50pt,-26pt);
	\node[anchor=center] at (47pt,-21.5pt) {\small $0$};
	\draw [fill=white] (84pt,-2pt) rectangle (90pt,-11pt);
	\node[anchor=center] at (87pt,-6.5pt) {\small $0$};
	\draw [fill=white] (84pt,33pt) rectangle (90pt,24pt);
	\node[anchor=center] at (87pt,28.5pt) {\small $0$};
	\draw [fill=white] (84pt,-37pt) rectangle (90pt,-46pt);
	\node[anchor=center] at (87pt,-41.5pt) {\small $0$};
	
\end{tikzpicture}
	\vspace{7pt}
	\caption{Dollar auction game in financial systems}
	\label{fig:auction}
\endminipage\hfill
\hspace{0.01\textwidth}
\end{figure}

\subsection{Dollar auction}

We also show an example of the dollar auction game \cite{dollar} in financial systems. Consider the system in Figure \ref{fig:auction}, and assume that banks $u'$ and $v'$ want to influence this system by donating extra funds to banks $u$ or $v$ (as in Theorem \ref{th:injectu}). Note that the payoff of $u'$ and $v'$ depends on the recovery rates of $u$ and $v$, respectively, which in turn have a recovery rate depending on each other. Due to the design of the system, $u'$ prefers bank $u$ to be in default, and thus it wants to increase $e_v$; similarly, $v'$ prefers bank $v$ to be in default, so it wants to increase $e_u$. We assume that 1 unit of money is a very high amount in our context, and thus $u'$ and $v'$ cannot donate enough to ensure that $u$ or $v$ pays its debt entirely from external assets; i.e. we assume that $e_u,e_v<1$ even after the donation of extra funds.

For a convenient analysis, we assume that there is a small minimum amount $\epsilon$ of funds that $u'$ or $v'$ can donate in one step. In our example, we choose a $\delta$ value in the magnitude of this $\epsilon$, e.g. $\delta=6 \epsilon$.

Let us now analyze the recovery rates of $u$ and $v$ in the solutions of the system.

\begin{itemize}
 \item The vector $r_u=r_v=1$ cannot be a solution, since it would imply no payment on the incoming CDSs, and thus these recovery rates would only be possible if $e_u, e_v \geq 1$.
 \item If a vector $r_u=1$, $r_v<1$ is a solution, then since $v$ receives no incoming payments, we must have $r_v=\frac{e_v}{1}=e_v$. Thus bank $u$ has assets of $e_u+1-e_v$, which has to be at least 1 for $r_u=1$ to hold. Hence this is only a solution if $e_u+1-e_v \geq 1$, i.e. $e_u \geq e_v$. In a symmetric manner, $r_v=1$, $r_u=e_u$ is only a solution if $e_v \geq e_u$.
 \item If $r_u<1$, $r_v<1$ in a solution, then $r_u=e_u+1-r_v$ and $r_v=e_v+1-r_u$ must hold. This implies $e_u=e_v$, and $r_u+r_v=1+e_u$. Hence if $e_u=e_v$, then any $r_u, r_v$ with $r_u+r_v=1+e_u$ provides a solution.
\end{itemize}

\noindent Thus as long as $e_u,e_v<1$, the behavior of the system is as follows:

\begin{itemize}
 \item If $e_u < e_v$, then the only solution is $r_u=e_u$, $r_v=1$. This means $q_{u'}=\delta \cdot (1-e_u)$ and $q_{v'}=0$.
 \item If $e_u > e_v$, then the only solution is $r_u=1$, $r_v=e_v$. This implies $q_{u'}=0$ and $q_{v'}=\delta \cdot (1-e_v)$.
 \item If $e_u = e_v$, then any $r_u,r_v \leq 1$ with $r_u+r_v=1+e_u$ is a solution of the system. In the general case, $q_{u'}=\delta \cdot (1-r_u)$ and $q_{v'}=\delta \cdot (1-r_v)$.
\end{itemize}

This describes a setting that is very similar to a dollar auction. In the beginning, with $e_u = e_v = 0$, we have a range of different solutions, and a choice among these depends on a financial authority. One of the banks (say, bank $u'$) decides to donate a small $\epsilon$ amount of funds to $v$; then with $e_v=\epsilon > e_u=0$, bank $u'$ receives a payment of $\delta \cdot (1-0)$ in the only resulting solution. At this point, the payoff of $v'$ is 0; however, at the cost of donating $2 \cdot \epsilon$ funds to $u$, it could achieve $e_u=2 \epsilon > e_v=\epsilon$, thus resulting in a single solution with a payoff of $q_{v'}=\delta \cdot (1-\epsilon)$. Since this increases the payoff of $v'$ by $\delta \cdot (1-\epsilon)$ at the cost of only $2\epsilon$, this is indeed a rational step for the appropriate $\delta$ and $\epsilon$ values. However, then $u'$ is again motivated to donate $2 \epsilon$ more funds to increase $e_v$ over $e_u$ again, and so on.

Assuming that both $u'$ and $v'$ has at most $\frac{1}{2}$ funds to donate, we always have $e_u,e_v \in [0, \frac{1}{2}]$. This shows that e.g. if we have $e_u > e_v$, then the payoff of bank $v'$ is always within
\[ q_{v'}=\delta \cdot (1-e_v) \; \in \; [\delta / 2\,, \, \delta]=[3 \epsilon\,,\, 6 \epsilon]. \]
Hence in every step, it is indeed rational for $v'$ to donate another $2 \epsilon$ funds, since it increases its payoff from 0 to at least $3 \epsilon$. After a couple of rounds, $u'$ and $v'$ will have both donated significantly more money than their payoff of at most $6 \epsilon$. However, the banks are still always tempted to execute the next donation step to mitigate their losses.

\newpage

\bibliography{references}

\begin{appendices}

\section{Further two-player games in financial systems} \label{App:games}

We now show some further examples of financial system that represent other well-known two-player-two-strategy games, similarly to the case of the prisoner's dilemma.

\subsection{Stag Hunt}

We first analyze the financial system in Figure \ref{fig:stag}, which represents the coordination game known as stag hunt \cite{generalGT}. We again assume that the two acting nodes $v_1$ and $v_2$ can only execute the action of completely removing their incoming debt contract from $u_1$ and $u_2$, respectively. As before, we refer to the decisions of removing and not removing the debt as cooperation and defection, respectively.

Recall that canceling an incoming debt and donating funds to another bank are very similar operations in some sense. With a slight modification to our system, we could also present the same example game in a setting where the acting banks must decide to donate or not donate a specific amount of funds to a bank. For our example systems, we select the action that allows a simpler presentation.

Let us analyze the payoffs in the different strategy profiles. If both players cooperate, then both $u_1$ and $u_2$ will only have a liability of 2, which implies $r_{u_1}=r_{u_2}=1$. In this case, $w$ receives no payment from either of the CDSs, resulting in $r_w=0$. This means that both $v_1$ and $v_2$ get a payment of $3$ from their incoming CDSs. With their debt contracts canceled, we get $q_{v_1}=q_{v_2}=3$.

If both players defect and keep their debt contract, then both $u_1$ and $u_2$ will have a recovery rate of only $\frac{1}{2}$. This implies a payment of $1$ to $w$ on both CDSs, so $w$ avoids default with $r_w=1$. This means that the acting nodes will not receive any payment on the CDS. On their debt contracts, they both receive $\frac{1}{2} \cdot 2$, i.e. $q_{v_1}=q_{v_2}=1$.

Finally, assume that $v_1$ cooperates but $v_2$ defects. In this case, we end up with recovery rates of $r_{u_1}=1$ and $r_{u_2}=\frac{1}{2}$. Thus $w$ only receives payment on the CDS that is in reference to $u_2$. However, this payment of $\frac{1}{2} \cdot 2$ is already enough for $w$ to fulfill its liabilities, and hence $r_w=1$. Again, $v_1$ and $v_2$ do not receive any payment on the CDS. However, $v_2$ still has an incoming debt contract that ensures a payment of $\frac{1}{2} \cdot 2 = 1$, while $v_1$ has no assets at all. Thus the solution provides $q_{v_1}=0$ and $q_{v_2}=1$. In a symmetric manner, the case when $v_2$ cooperates and $v_1$ defects incurs $q_{v_1}=1$, $q_{v_2}=0$.

Thus the system represents a game where the players are incentivized to coordinate their strategies. Both the case when both banks cooperate and when both banks defect is a pure Nash-Equilibrium, with mutual cooperation being the social optimum. However, if a bank is unsure whether the other bank will cooperate, it might be motivated to defect in order to avoid the risk of getting no payoff at all.

\subsection{Chicken game}

Finally, we also provide an example of the chicken game (also known as the hawk-dove game \cite{generalGT}) when the pure Nash-Equilibria are obtained in the asymmetric strategy profiles.

Consider the financial system in Figure \ref{fig:chicken}, and assume the acting banks $v_1$ and $v_2$ now have the options to either donate $1$ unit or money to another bank, or do not donate money at all. Due to the structure of the network, the nodes are motivated to donate this $1$ unit of money to $u$, since this results in a payment on their incoming CDS contract. We again refer to donating a unit of money to $u$ as cooperation, and not donating as defection.

If both nodes defect, then $u$ still has no assets at all, implying $r_u=0$. This results in $r_w=1$, and hence the acting nodes receive no incoming payment, so $q_{v_1}=q_{v_2}=0$.

If both nodes cooperate, then $u$ has more than enough assets to pay its liabilities, resulting in $r_u=1$ and $r_w=0$. This means that both nodes get a payment of $3$ in the CDS. After subtracting the amount they have donated, we get $q_{v_1}=q_{v_2}=2$.

However, to ensure that $u$ does not go into default, it is enough if only one of the two nodes make a donation. I.e. if $v_1$ cooperates but $v_2$ defects, then $u$ still has $1$ asset, which already implies $r_u=1$, $r_w=0$ and a payment of $3$ to both $v_1$ and $v_2$ on their incoming CDS. After subtracting the donated funds, this gives $q_{v_1}=2$ and $q_{v_2}=3$. Similarly, if $v_2$ cooperates and $v_1$ defects, we obtain $q_{v_1}=3$, $q_{v_2}=2$.

The payoffs show that there is no dominant strategy in the game: if $v_1$ cooperates, then the best response of $v_2$ is to defect, while if $v_1$ defects, then the best response of $v_2$ is to cooperate. This implies that the two pure Nash-Equilibria are obtained in the strategy profiles when the banks choose the opposite strategies. 

Note that we can easily generalize this setting to the case of more than 2 acting nodes, resulting in the so-called volunteer's dilemma. For any $k$, we can add distinct banks $v_1,v_2, ...,v_k$ that are all connected to the financial network in the same way (through an incoming CDS of weight $3$ in reference to $w$), and all have the same two options of either donating $1$ unit of money to $u$ or not acting at all. Note that we also have to ensure that the (currently unlabeled) debtor of the CDSs to these acting nodes has enough resources to make payments on these CDSs in any case, i.e. it must have external assets of at least $3 \cdot k$.

In this case, we obtain a game where again only one volunteer bank $v_i$ is required to make a donation to $u$, and this already ensures a payoff of $3$ for every other bank (and a payoff of $2$ for $v_i$). In this game, the pure Nash-Equilibria are the strategy profiles where exactly one bank cooperates, and the remaining banks all defect.

\begin{figure}
\centering
\captionsetup{justification=centering}
\hspace{0.02\textwidth}
\minipage{0.4\textwidth}
	\centering

\begin{tikzpicture}
	
	\draw[very thick, blue, arrows=-latex] (25pt,100pt) -- (2pt,35.5pt);
	\draw[very thick, blue, arrows=-latex] (35pt,60pt) -- (4pt,34pt);
	\draw[very thick, blue, arrows=-latex] (35pt,60pt) -- (84pt,60pt);
	\draw[very thick, blue, arrows=-latex] (25pt,100pt) -- (84pt,100pt);
	\draw[very thick, blue, arrows=-latex] (60pt,30pt) -- (86pt,4pt);
	\draw[very thick, brown, arrows=-latex] (0pt,28pt) -- (54pt,28pt);
	\draw[very thick, brown, arrows=-latex] (0pt,32pt) -- (54pt,32pt);
	\draw[very thick, brown, arrows=-latex] (90pt,0pt) -- (90pt,54pt);
	\draw[very thick, brown, arrows=-latex] (90pt,0pt) -- (105pt,15pt) -- (105pt,60pt) -- (90pt, 75pt) -- (90pt, 94pt);
	\draw[very thick, brown, densely dotted] (90pt,29pt) -- (60pt,29pt);
	\draw[very thick, brown, densely dotted] (105pt,33pt) -- (60pt,33pt);
	\draw[very thick, brown, densely dotted] (25pt,32pt) -- (25pt,100pt);
	\draw[very thick, brown, densely dotted] (35pt,28pt) -- (35pt,60pt);
	
	\node[anchor=center] at (30pt,21.5pt) {\normalsize $2$};
	\node[anchor=center] at (30pt,38pt) {\normalsize $2$};
	\node[anchor=center] at (95pt,27pt) {\normalsize $3$};
	\node[anchor=center] at (110pt,33pt) {\normalsize $3$};
	\node[anchor=center] at (55pt,106pt) {\normalsize $2$};
	\node[anchor=center] at (62pt,66pt) {\normalsize $2$};
	\node[anchor=center] at (9pt,69pt) {\normalsize $2$};
	\node[anchor=center] at (16pt,51pt) {\normalsize $2$};
	\node[anchor=center] at (72pt,10pt) {\normalsize $1$};
	
	\draw[black, fill=white] (0pt,30pt) circle (1.7ex);
	\draw[black, fill=white] (60pt,30pt) circle (1.7ex);
	\draw[black, fill=white] (25pt,100pt) circle (1.7ex);
	\draw[black, fill=white] (35pt,60pt) circle (1.7ex);
	\draw[black, fill=white] (90pt,60pt) circle (1.7ex);
	\draw[black, fill=white] (90pt,100pt) circle (1.7ex);
	\draw[black, fill=white] (90pt,0pt) circle (1.7ex);
	
	\node[anchor=center] at (60pt,30pt) {\normalsize $w$};
	\node[anchor=center] at (25.5pt,99.5pt) {\normalsize $u_1$};
	\node[anchor=center] at (35.5pt,59.5pt) {\normalsize $u_2$};
	\node[anchor=center] at (90.5pt,59.5pt) {\normalsize $v_2$};
	\node[anchor=center] at (90.5pt,99.5pt) {\normalsize $v_1$};
	
	\draw [fill=white] (4pt,28pt) rectangle (10pt,19pt);
	\node[anchor=center] at (7pt,23.5pt) {\small $4$};
	\draw [fill=white] (64pt,28pt) rectangle (70pt,19pt);
	\node[anchor=center] at (67pt,23.5pt) {\small $0$};
	\draw [fill=white] (29.5pt,97.5pt) rectangle (35.5pt,88.5pt);
	\node[anchor=center] at (32.5pt,93pt) {\small $2$};
	\draw [fill=white] (39.5pt,57.5pt) rectangle (45.5pt,48.5pt);
	\node[anchor=center] at (42.5pt,53pt) {\small $2$};
	\draw [fill=white] (94.5pt,57.5pt) rectangle (100.5pt,48.5pt);
	\node[anchor=center] at (97.5pt,53pt) {\small $0$};
	\draw [fill=white] (94.5pt,97.5pt) rectangle (100.5pt,88.5pt);
	\node[anchor=center] at (97.5pt,93pt) {\small $0$};
	\draw [fill=white] (94pt,-2pt) rectangle (100pt,-11pt);
	\node[anchor=center] at (97pt,-6.5pt) {\small $6$};
	
\end{tikzpicture}
	\vspace{-5pt}
	\caption{Stag hunt game in a financial system}
	\label{fig:stag}
\endminipage\hfill
\hspace{0.09\textwidth}
\minipage{0.4\textwidth}	
  \centering
	\vspace{5pt}

\begin{tikzpicture}
	
	\draw[very thick, blue, arrows=-latex] (30pt,80pt) -- (3pt,35pt);
	\draw[very thick, blue, arrows=-latex] (60pt,30pt) -- (86pt,4pt);
	\draw[very thick, brown, arrows=-latex] (0pt,30pt) -- (54pt,30pt);
	\draw[very thick, brown, arrows=-latex] (90pt,0pt) -- (90pt,54pt);
	\draw[very thick, brown, arrows=-latex] (90pt,0pt) -- (105pt,15pt) -- (105pt,60pt) -- (90pt, 75pt) -- (90pt, 94pt);
	\draw[very thick, brown, densely dotted] (90pt,29pt) -- (60pt,29pt);
	\draw[very thick, brown, densely dotted] (105pt,33pt) -- (60pt,33pt);
	\draw[very thick, brown, densely dotted] (30pt,30pt) -- (30pt,80pt);
	
	\node[anchor=center] at (30pt,23pt) {\normalsize $1$};
	\node[anchor=center] at (95pt,27pt) {\normalsize $3$};
	\node[anchor=center] at (110pt,33pt) {\normalsize $3$};
	\node[anchor=center] at (11.5pt,59pt) {\normalsize $1$};
	\node[anchor=center] at (72pt,10pt) {\normalsize $1$};
	
	\draw[black, fill=white] (0pt,30pt) circle (1.7ex);
	\draw[black, fill=white] (60pt,30pt) circle (1.7ex);
	\draw[black, fill=white] (30pt,80pt) circle (1.7ex);
	\draw[black, fill=white] (90pt,60pt) circle (1.7ex);
	\draw[black, fill=white] (90pt,100pt) circle (1.7ex);
	\draw[black, fill=white] (90pt,0pt) circle (1.7ex);
	
	\node[anchor=center] at (60pt,30pt) {\normalsize $w$};
	\node[anchor=center] at (30pt,80pt) {\normalsize $u$};
	\node[anchor=center] at (90.5pt,59.5pt) {\normalsize $v_2$};
	\node[anchor=center] at (90.5pt,99.5pt) {\normalsize $v_1$};
	
	\draw [fill=white] (4pt,28pt) rectangle (10pt,19pt);
	\node[anchor=center] at (7pt,23.5pt) {\small $1$};
	\draw [fill=white] (64pt,28pt) rectangle (70pt,19pt);
	\node[anchor=center] at (67pt,23.5pt) {\small $0$};
	\draw [fill=white] (34pt,78pt) rectangle (40pt,69pt);
	\node[anchor=center] at (37pt,73.5pt) {\small $0$};
	\draw [fill=white] (94.5pt,57.5pt) rectangle (100.5pt,48.5pt);
	\node[anchor=center] at (97.5pt,53pt) {\small $0$};
	\draw [fill=white] (94.5pt,97.5pt) rectangle (100.5pt,88.5pt);
	\node[anchor=center] at (97.5pt,93pt) {\small $0$};
	\draw [fill=white] (94pt,-2pt) rectangle (100pt,-11pt);
	\node[anchor=center] at (97pt,-6.5pt) {\small $6$};
	
\end{tikzpicture}
	\vspace{-5pt}
	\caption{Chicken game in a financial system}
	\label{fig:chicken}
\endminipage\hfill
\hspace{0.01\textwidth}
\end{figure}

\end{appendices}

\end{document}